\documentclass{amsart}
\usepackage{amsmath,amsfonts,amssymb,amscd,pb-diagram,mathrsfs,amsthm,graphicx,color}
\usepackage{cite, hyperref}
\def\sE{\mathscr{E}}
\def\sA{\mathscr{A}}
\def\sC{\mathscr{C}}

\def\sH{\mathscr{H}}
\def\sG{\mathscr{G}}
\def\sV{\mathscr{V}}
\def\u{\underline}
\def\CC{\mathbb{C}}
\def\RR{\mathbb{R}}
\def\QQ{\mathbb{Q}}
\def\NN{\mathbb{N}}
\def\ZZ{\mathbb{Z}}
\def\OO{\mathscr{O}}
\def\Set{\mathbf{Set}}
\def\fSet{\mathbf{fSet}}
\def\Bool{\mathbf{Bool}}
\def\fBool{\mathbf{fBool}}
\def\Conv{\mathbf{Conv}}
\def\Ord{\mathbf{Ord}}
\def\dLat{\mathbf{dLat}}
\def\Frm{\mathbf{Frm}}
\def\Loc{\mathbf{Loc}}
\def\FrmTop{\mathbf{FrmTop}}
\def\bFrmTop{\mathbf{bFrmTop}}
\def\lFrmTop{\mathbf{lFrmTop}}
\def\Dcpo{\mathbf{Dcpo}}

\def\dn{\downarrow}
\def\up{\uparrow}

\def\sF{\mathscr{F}}
\def\sK{\mathscr{K}}
\def\TT{\mathbf{T}}
\def\VV{\mathbf{V}}

\def\XX{\mathscr{X}}
\def\IIdl{\mathbf{Idl}}
\def\BBox{\mathbf{Box}}
\DeclareMathOperator{\Idl}{\mathrm{Idl}}
\DeclareMathOperator{\Sub}{\mathrm{Sub}}
\DeclareMathOperator{\Hom}{\mathrm{Hom}}
\DeclareMathOperator{\colim}{\mathrm{colim}}
\DeclareMathOperator{\id}{\mathrm{id}}
\DeclareMathOperator{\pr}{\mathrm{pr}}

\DeclareMathOperator{\Sh}{\mathrm{Sh}}
\DeclareMathOperator{\Sp}{\mathrm{Sp}}
\def\op{\mathrm{op}}
\theoremstyle{plain}
\newtheorem{lem}{Lemma}
\newtheorem{pro}{Proposition}
\newtheorem{thm}{Theorem}
\theoremstyle{definition}
\newtheorem{defn}{Definition}
\title{Non-signalling boxes and Bohrification}
\author{Jan Gutt}
\author{Marek Ku\'s}
\address{Center for Theoretical Physics, Polish Academy of Sciences, Al. Lotnik\'ow 32/46, 02-668 Warszawa}
\begin{document}
\maketitle
\sloppy
\section{Introduction}
\subsection{}
The premise of this note is the following observation: \emph{the
formalism of Bohrification, as developed by Heunen et
al.~\cite{heunen-landsman-spitters}, is a natural place for the
interpretation of general non-signalling theories}.
The former, i.e. Bohrification, is an embodiment of
of Bohr's idea that the account of all
evidence concerning quantum phenomena, despite their non-classical
character, must be expressed in classical terms \cite{bohr49}.
The latter, i.e. non-signalling theories, are extensions of the 
notion of a \emph{probability theory} of a physical system,
allowing correlations forbidden in a quantum mechanical model (cf. ~\cite{tsirelson}).
We show that the condition of non-signalling allows one to present these 
extended probability theories in the same terms as `Bohrified' quantum systems.

In standard quantum
theory~\cite{birkhoff-neumann} the logic of a system described by a
Hilbert space $\sH$ is represented by the orthomodular lattice of
closed subspaces in $\sH$. The involution sending a subspace to its
orthogonal complement represents logical negation, satisfying the law
of excluded middle: measuring the spin of an electron will yield either
`up' or `down', \emph{tertium non datur}. On the other hand, the
lattice is non-distributive: $x$-spin up does not imply $x$-spin up and
$z$-spin up \emph{or} $x$-spin up and $z$-spin down (incompatibility of
the two measurements is reflected by non-distributivity of the
sub-lattice they `generate', just as by non-commutativity of the
corresponding sub-algebra of operators). Having the lattice stand for
the logic of the system, one derives its probability theory where
\emph{states} assign `probabilities' to elements of the lattice,
respecting the underlying structure (order and complementation). These
states turn out to coincide with the usual density matrices by a
celebrated theorem of Gleason (as long as $\dim\sH\ge3$,~\cite{gleason}).

One of the original purposes of the Bohrification programme has been to
find an alternative logical foundation for the orthodox quantum
theory, replacing the non-distributive orthomodular lattices
with the distributive
logic of a point-free space. From the classical perspective,
distributivity comes at the price of the absence of the law of excluded
middle; from the intuitionistic (or constructive) viewpoint this is
however a feature rather than a flaw, and a characteristic of the true
logic of physical observation (see~\cite{vickers}).

On the other hand, non-signalling theories embodied in so-called `box
worlds' are an attempt at extending quantum mechanics as a theory of
\emph{probability}, reaching beyond the well-known bounds imposed on
correlations by the orthodox theory. We shall see that not only are
such hypothetical `box worlds' amenable to the process of
Bohrification, but furthermore they acquire a natural logical
structure, compatible with their probabilistic content. In particular,
probability valuations on the logic thus associated with a box world
are naturally identified with the standard `non-signalling box states'
(Theorem \ref{thm:states}). Moreover, one finds that a more general
class of non-signalling theories may be identified with their
`Bohrified' models, resulting in a common representation encompassing
box worlds, orthodox quantum systems and a potentially interesting
in-between (cf. \ref{ss:general}).

\subsection{Bohrification}
In certain sense, the goal of Bohrification is to find a well-behaved
\emph{phase space} for a quantum system. This is achieved by
constructing its logical avatar: a frame, i.e. a complete distributive
lattice where finite meets distribute over arbitrary joins. The lattice
of open subsets of a given topological space provides an example of a
frame, and in a leap of abstraction one may view any frame as a
`pointless topology' -- i.e., a virtual space examined only through its
collection of open subsets and the lattice operations on them. This
does seem to correspond to the way a physicist observes the phase space
of a system -- and thus to the actual logic of such observations.
To emphasise this
interpretation, one defines the category of \emph{locales} as the
opposite of the category of frames (recall that the topology functor
from topological spaces to frames is contravariant).
The reference to \emph{open subspaces}
reflects the intuitionism of the logic: negation corresponds to taking
the \emph{interior} of the complement, whence the disjunction of a
proposition and its negation need not be true.

However, we have not yet revealed a crucial technical aspect. Actually,
the new intuitionistic logic of the \emph{quantum} system is realised
not as a frame (or any paritally ordered \emph{set}), but rather as a
\emph{frame object} in a suitable \emph{topos}, intrinsically
associated with the system under consideration. Thus, the meta-logic
describing the logic of the system is the internal logic of a topos.
For instance, probability valuations on the frame object are viewed as
morphisms into a real numbers object (however, since the internal logic
is intuitionistic, different constructions of the real numbers that
would classically yield the same \emph{set}, may lead to non-isomorphic
objects when interpreted internally: e.g., one-sided Dedekind cuts are
in general not equivalent to two-sided ones). Dually, the frame object
is seen as an internal locale: the `phase space object'. Since the
topos arising in Bohrification is simply that of sheaves on some base
locale, one may represent the phase space \emph{externally} by a locale
(now in the category of sets) over the base locale -- even more
tangibly, these external locales are in fact topological spaces.

The original construction of Heunen et
al.~\cite{heunen-landsman-spitters} is set in the framework of
$C^*$-algebras, applicable to field-theoretic or statistical-mechanical
systems with infinitely many degrees of freedom. Recovering the
`correct' topological structure of the phase space is then a delicate
matter. Given a system described by a $C^*$-algebra $\sA$, the authors
consider the set $\sC$ of abelian $C^*$-subalgebras of $\sA$, partially
ordered by inclusion. Viewing $\sC$ as a category (with arrows
expressing the order relation), we have a tautological functor $\u\sA$
from $\sC$ to commutative $C^*$-algebras, sending $A \in \sC$ to $A$
iself. We may then view $\u\sA$ as an object of the presheaf topos
$\Set^\sC$ carrying the structure of an \emph{internal commutative
$C^*$-algebra} (the presheaf topos has functors $\sC \to \Set$ as
objects, and natural transformations as morphisms; equivalently, its
objects are $\sC$-shaped diagrams of sets). Now, an internal
\emph{Gelfand duality} is invoked to produce an internal locale $\u X$
such that $C(\u X, \u\CC) \simeq \u\sA$ (where $\u\CC$ is a suitably
defined internal locale of complex numbers, and $C(-,-)$ is the object
of continuous maps between a pair of internal locales). This $\u X$ is
the \emph{internal phase space} of the system, and the corresponding
internal frame is its logic.

For a finite-level system described by a finite-dimensional Hilbert
space $\sH$ the construction of the internal frame is much more direct
(see~\cite{caspers-heunen-landsman-spitters}). Indeed, setting $\sA =
B(\sH)$, i.e. linear operators on $\sH$,
and proceeding as before, one may consider the
internal Boolean algebra $P(\u\sA)$ of projections: it is the functor
sending $A \in \sC$ to the Boolean algebra $P(A)$ of projections in
$A$. This is not a frame yet, as $P(\u\sA)$ is not complete (even
though its values are finite Boolean algebras, thus complete,
$P(\u\sA)$ itself is only so-called $\tilde K$-finite, which does not
suffice to establish completeness internally). One thus applies the
operation of internal \emph{ideal completion}, obtaining a frame $\Idl
P(\u\sA) \supset P(\u\sA)$ defined by a suitable universal property.
The internal phase space is the corresponding locale (in fact, the
passage from the internal Boolean algebra $P(\u\sA)$ to the latter
locale is an internal version of the \emph{Stone spectrum}). An
important result, relying on Gleason's theorem~\cite{gleason}, is that
internal probability valuations on $\Idl P(\u\sA)$ are in a natural
one-to-one correspondence with density matrices on $\sH$ for $\dim \sH
\ge3$.

Of course, the above description may be restated without
any reference to $B(\sH)$. Namely, one lets $\sC$ be set of \emph{orthogonal decompositions}
of $\sH$, partially ordered by refinement, and then considers the functor $L : \sC \to \Bool$
sending a given orthogonal decomposition $\sH = \bigoplus_{i\in I} V_i$
to $2^I$; a refinement of $(V_i)_{i\in I}$ to $(V_j)_{j\in J}$
is sent to the homomorphism $2^I \to 2^J$ induced by the inclusion map $J \to I$.
Now, $L$ is isomorphic to the previous Boolean algebra $P(\u\sA)$ and
its internal ideal completion $\Idl L$ corresponds to the internal phase space.
\label{ss:keyin}
The key insight we take from this discussion is: (1) $\sC$ is the set of \emph{compatible
measurement contexts}, partially ordered by refinement of available information;
(2) the Boolean algebras describing the \emph{logic of measurements} in each context
give rise to an internal Boolean algebra in $\Set^\sC$; (3) the internal ideal completion
of the latter corresponds to the internal phase space of the general system subject to
our measurement arrangement; (4) the states of the system are identified with internal probability
valuations.

\subsection{Box worlds}
Shifting the view away from logic, one recognizes the probabilistic
essence of quantum mechanics as a theory prescribing,
for a given collection of admissible measurements, the convex
set of states assigning a `probability' to each measurement
and each outcome. A partial `conjunction' allows for performing
several \emph{compatible} measurements simultaneously, and
combining their outcomes. The classical scenario illustrating
these notions involves two parties in causally separated laboratories,
each equipped with a pair
of binary measurements to be performed on a shared distributed
system. In course of the experiment,
each party subjects the system to one of the
measurements at its disposal, and records the outcome.
Assuming the initial state of the system may be consistently
reproduced, the parties repeat the experiment and eventually
communicate to assemble a list of `empirical probabilities'
for each combined outcome. The assumption that the parties
be causally separated guarantees that the measurement
chosen by the first party is always compatible with that
chosen by the second party.

The \emph{non-local} correlations
between the results of measurements in the two laboratories
distinguish a non-classical theory (more precisely,
non-classical states) from a classical one.
One expression detecting non-locality forms the
CHSH inequality:
$$
\langle a_1 b_1\rangle + \langle a_1 b_2 \rangle
+ \langle a_2 b_1 \rangle - \langle a_2 b_2 \rangle
\le 2.
$$
Here we encode the binary measurements
in terms of $\pm1$-valued observables:
$a_1, a_2$ in
the first laboratory, $b_1, b_2$ in the second one;
$\langle \cdot \rangle$
denotes the expectation in a given state $P$:
$$
\langle a_i b_j \rangle = \sum_{\alpha,\beta = \pm1}
P(a_i=\alpha \wedge b_j=\beta)\alpha\beta.$$
If the inequality
is satisfied, the state is consistent with a classical theory;
otherwise, it describes a genuinely non-classical situation.
Interestingly, the states described by standard quantum
mechanics \emph{do not} saturate the obvious algebraic
upper bound for the CHSH expression, i.e. $4$: instead,
they satisfy an upper bound of $2\sqrt2$~\cite{tsirelson}.

One may ponder whether this is indeed the bound chosen
by nature. Seeking an axiomatics for general `probabilities'
of the form $P(a_i=\alpha \wedge b_j = \beta)$, one would
certainly like to ensure that the one-party
expectation $\langle a_i\rangle$,
or the probability $P(a_i=\alpha)$, may be consistently inferred
as a marginal.
That is, that
$$
P(a_i=\alpha \wedge b_1=1) +
P(a_i=\alpha \wedge b_1=-1) =
P(a_i=\alpha \wedge b_2=1) +
P(a_i=\alpha \wedge b_2=-1)
$$
for all $i=1,2$ and $\alpha=\pm1$: a property
referred to as \emph{non-signalling} (for otherwise
the first laboratory would obtain information
on the \emph{choice} of a measurement at the second laboratory,
violating causality). Popescu and Rohrlich~\cite{popescu-rohrlich} provided the first
example of a non-signalling state $P$ for which
the CHSH expression achieves its algebraic maximum of $4$.
Their construction, referred to as the \emph{PR box},
has developed into the study of `non-signalling boxes' or `box
worlds', modelling \emph{super-quantum}
correlations.

\subsection{}\label{ss:prelim-pr}
Let us attempt a preliminary `Bohrification' of the above box-world. Following
the `key insight' of \ref{ss:keyin}, we first identify the partially ordered set
$\sC$ of compatible measurement contexts. As indicated in the description of
the measurement protocol, each party selects one measurement out of two, whence
we obtain four contexts represented by subsets
$$ \{ a_1,b_1\},\ \{ a_1, b_2\},\ \{ a_2, b_1\},\ \{ a_2,b_2\}\ \subset\ \{a_1,a_2,b_1,b_2\}. $$
Furthermore, since we have assumed that expectations of the form $\langle a_i\rangle$,
$\langle b_j\rangle$ may be \emph{consistently} obtained as marginals, we may add another
four contexts
$$ \{ a_1\}, \{a_2\}, \{b_1\}, \{b_2\}\ \subset\ \{a_1,a_2,b_1,b_2\}. $$
These might be alternatively interpreted as corresponding to a situation where either (1) only one
party performs a measurement, or (2) both parties perform measurements, but one of them discards
the result. Finally, we may for completeness add the empty context $\emptyset$, corresponding to
either no measurements being taken, or all results being discarded. It follows that the partially
ordered set $\sC$ of measurement contexts is a subset of the power-set $2^{\{a_1,a_2,b_1,b_2\}}$
consisting of subsets whose intersection with $\{a_1,a_2\}$ and $\{b_1,b_2\}$ is of cardinality
at most one. In particular, viewed as a category, $\sC$ is a sub-category of $\Set$.

We now need to construct a functor assigning to each context $S \in \sC$ a
Boolean algebra representing the measurement logic in that context. Since $S$ is a set
of compatible binary measurements, it is immediate that the set of possible \emph{outcomes}
if $2^S$. The logic of measurements is the power-set of the set of outcomes, in our case
$2^{2^S}$. Equivalently, it is simply the \emph{free Boolean algebra} on $S$. Hence, the
desired functor $L : \sC \to \Bool$ is simply a restriction of the free functor
$\Set \to \Bool$ to $\sC \subset \Set$ (in particular, for $S \le S'$ in $\sC$ we
obtain a homomorphism $L(S \to S') : LS \to LS'$ of measurement logics, induced by
the restriction map $2^{S} \to 2^{S'}$ on measurement outcomes). Viewed as an object
of $\Set^\sC$, the functor $L$ is an internal Boolean algebra.

We may also describe the internal frame $\Idl L$ corresponding to the internal phase space.
For now, it will be enough to look at its \emph{global sections}, or alternatively at its
value at $\emptyset \in \sC$. This is the set of \emph{sub-functors} $I \subset L$ such that
for each $S \in \sC$, the subset $I(S) \subset L(S)$ is an order ideal. Since each
$L(S)$ is finite, $I(S)$ is necessarily principal, i.e. a \emph{lower} set $\dn x$
for some $x \in L(S)$. Thus, we may identify $(\Idl L)(\emptyset)$ with the set of
maps $\xi : \sC \to \coprod_{S\in\sC} L(S)$ such that $\xi(S) \in L(S)$ for all
$S \in \sC$ and $L(S\to S')\xi(S) \le \xi(S')$ whenever $S \le S'$ in $\sC$. Alternatively,
$(\Idl L)(\emptyset)$ is the set of all \emph{upper} subsets of $X = \coprod_{S\in \sC} 2^S$.
Here $X$ is viewed as the set of \emph{measurement outcomes}, fibred over
the set of measurement contexts $\sC$, and partially ordered by refinement
($f' : S' \to 2$ refines $f : S \to 2$ if $S \subset S'$ and $f'|_S = f$). We may thus finally
say that $(\Idl L)(\emptyset)$ is the \emph{frame} $\OO(X)$ of open subsets in $X$ equipped with
the Alexandrov topology. Accordingly, $X$ is the \emph{external phase space} of our box world
(this description of the external phase space is taken from Heunen et al.,~\cite{heunen-landsman-spitters}).

\section{Categorical preliminaries}

\subsection{}
The logical structure of $L$, $\Idl L$ and lastly $X$ may be interesting, but it may only be
justified by its compatibility with the probabilistic content of the box world. We will show
in the next section that probability valuations on $\Idl L$ (or on $L$) are indeed in
one-to-one correspondence with non-signalling box-world states. For that purpose, we will
need to make these notions more precise.

\subsection{} Recall that a category $\sE$ (for convenience assumed to be locally small)
is a topos if it has finite limits and colimits,
exponential objects, and a sub-object classifier. Note that the existence
of finite limits and colimits implies the existence of a terminal object
(denoted $1_\sE$) and initial object (denoted $0_\sE$). The existence of exponential objects
means that for each object $X$ of $\sE$, the functor $X \times - : \sE \to \sE$
has a right adjoint $-^X : \sE \to \sE$. A sub-object classifier is an object
$\Omega$ representing the functor $\Sub : \sE^\op \to \Set$ of sub-objects. That is,
for each object $X$ of $\sE$ there is a natural bijection between $\Hom_\sE(X,\Omega)$
and the set of equivalence classes of monomorphisms into $X$. The sub-object classifier
is unique up to unique isomorphism.

The basic example of a topos is of course the category $\Set$ of sets, with the
two-point set $2$ as a sub-object classifier.
The topoi that arise in the context of Bohrification are still of a very simple kind. Namely,
given a partially ordered set $\sC$, we identify it with a category (its objects are the elements
of $\sC$, and its hom-sets are either empty or singletons, reflecting the order relation).
Then, we form the \emph{Kripke topos} $\Set^\sC$:
its objects are functors $\sC \to \Set$, and its morphisms are natural transformations
(equivalently, the objects may be viewed as $\sC$-shaped diagrams of sets). One generically
interprets $\sC$ as a set of `contexts' ordered by `refinement'; then, an object of
the corresponding Kripke topos is a `context-dependent set', transforming as the context
is being refined. Limits and colimits in $\Set^\sC$ are computed `point-wise', i.e.
independently for each context. The terminal (resp. initial) object is a functor sending each context
to a singleton (resp. empty set). A sub-object of a functor $F:\sC \to \Set$
is just a sub-functor $F'$, i.e. a subset $F'(c) \subset F(c)$ for each $c \in \sC$
such that $F'(c)$ maps to $F'(c')$ whenever $c \le c'$ in $\sC$.
It is then easy to see that a sub-object classifier is given by the functor
$\Omega : \sC \to \Set$ sending $c \in \sC$ to the set of upper subsets contained in
$\up c$.

\subsection{}
Given a pair of topoi $\sE$, $\sE'$, a \emph{geometric morphism} $f : \sE \to \sE'$
is a pair of functors $f_* : \sE \leftrightarrows \sE' : f^*$ such that
$f^*$ is left adjoint to $f_*$ and preserves finie limits
(the notation is of course reminiscent of
the direct and inverse image functors for sheaves). Such $f$ is furthermore \emph{essential}
if $f^*$ possesses a further left adjoint, denoted $f_!$. Every topos admits
at most a unique (up to equivalence) geometric morphism to $\Set$.

Geometric morphisms of Kripke topoi arise from isotone maps on the underlying posets.
Indeed, given an isotone map $f : \sC \to \sC'$ (i.e. a functor when
posets are viewed as categories), we define
a geometric morphism $f : \Set^\sC \to \Set^{\sC'}$ where, for functors
$F : \sC \to \Set$ and $F' : \sC' \to \Set$, we have
$f^*F' = F' \circ f$
and
$$(f_*F)(c') = \lim_{f(c)=c'} F(c)$$
for all $c' \in \sC'$ (the action of $f_*F$ on arrows of $\sC$ is easy to deduce). In fact,
$f$ is essential with
$$
(f_!F)(c') = \colim_{f(c)=c'} F(c).
$$
In particular, identifying $\Set$ with the Kripke topos over the one-element poset $\ast$,
we always have the essential geometric morphism $f:\Set^\sC \to \Set$ induced
by the unique map $f:\sC \to \ast$. Its components $f_*, f_! : \Set^\sC \to \Set$
compute, respectively, limits and colimits of functors, while $f^* : \Set \to \Set^\sC$
gives the constant functor on a given set.

\subsection{}
A \emph{partially ordered object} in $\sE$ is a pair $(X,\le_X)$,
where $X$ is an object together with a sub-object $(\le_X) \subset X \times X$
such that (1) the intersection of $(\le_X)$ with its `transpose' is the
diagonal $\Delta_X \subset X \times X$, and (2) $\pr_{13} : (\le_X) \times_X (\le_X) \to X \times X$
factors through $(\le_X)$,  where the fibre product is with respect to the
right projection from the left factor and left projection from the right factor. These
two conditions are equivalent to requiring that the representable functor
$\Hom(-,X) : \sE^\op \to \Set$ factor through the category $\Ord$ of partially
ordered sets, where the order relation on $\Hom(Y,X)$ is the preimage of
$\Hom(Y, \le_X)$ under the map $\Hom(Y,X)\times \Hom(Y,X) \to \Hom(Y,X\times X)$.
A morphism between
$X \to Y$ between partially ordered objects is \emph{isotone}
if it maps $\le_X$ into $\le_Y$.  We let $\Ord_\sE$ denote
the category of partially ordered objects in $\sE$ and isotone morphisms.
The power-object functor $\Omega^- : \sE^\op \to \sE$ factors naturally
through $\Ord_\sE$, where the ordering on $\Omega^X$ is given by
factorisation, i.e. inclusion, of sub-objects of $X$. In particular,
$\Omega$ itself is naturally a partially ordered object. In a Kripke topos
$\Set^\sC$, partially ordered objects are simply functors $\sC \to \Ord$.

The object $(X,\le_X)$ of $\Ord_\sE$ is an \emph{internal distributive lattice}
if the representable functor $\Hom(-,X) : \sE^\op \to \Ord$ factors through
the forgetful functor from the category $\dLat$ of bounded distributive lattices.
If that is the case, there exist morphisms $\wedge_X, \vee_X : X \times X \to X$ such that
for each $Y$, the induced maps $\Hom(Y,X) \times \Hom(Y,X) \to \Hom(Y,X)$ are the
meet and join in $\Hom(Y,X)$. There are also
morphisms $0_X, 1_X : 1_\sE \to X$ such that for each $Y$, the induced elements
in $\Hom(Y,X)$ are the bottom and top.
An isotone morphism $\phi : X \to Y$ between
distributive lattices is a \emph{distributive lattice homomorphism} if
$\vee_Y \circ (\phi\times\phi) = \phi \circ \wedge_X$
and likewise for $\wedge$, and if
$\phi \circ 0_X = 0_Y$ and likewise for $1$ (equivalently
if $\Hom(Z,\phi)$ is a distributive lattice homomorphism
for each $Z$). The category of internal distributive lattices
and their homomorphisms in $\sE$ is denoted $\dLat_\sE$. It is a sub-category
of $\Ord_\sE$. Again, the power-object functor factors naturally through $\dLat_\sE$ and in particular
$\Omega$ is itself an internal distributive lattice.
In a Kripke topos $\Set^\sC$, internal distributive lattices are simply functors
$\sC \to \dLat$.

Finally, $(X,\le_X)$ is an \emph{internal Boolean algebra}
if $\Hom(-,X) : \sE^\op \to \Ord$ factors through the category
$\Bool$ of Boolean algebras. That is, if it is an internal distributive
lattice together with an antitone negation/complementation morphism $\neg_X : X \to X$
such that $\vee_X \circ (\id_X \times \neg_X) \circ \Delta_X$
factors through $1_X$, and an analogous composite with $\wedge_X$
factors through $0_X$. Internal Boolean algebras form a full subcategory
$\Bool_\sE$ of $\dLat_\sE$. In a Kripke topos $\Set^\sC$ these are
simply functors $\sC \to \Bool$.

\subsection{}
As completeness is a higher-order notion, the definition of internal
frames (as well as internal dcpo) is somewhat more involved.
We use monadicity to keep the discussion explicit.
Recall that
a \emph{monad} on a category $\sK$ is a triple $\TT=(T,\mu,\eta)$
where $T$ is an endofunctor of $\sK$ together with natural transformations
$\mu : T^2 \to T$ (`multiplication'), $\eta : \id \to T$ (`unit') such that
$\mu\circ T\eta = T\id$ and $\mu\circ T\mu = \mu \circ \mu T$.
A  \emph{$\TT$-algebra}
is a pair $(X,\xi)$ where $X$ is an object, and
$\xi : TX \to X$ a morphism (`structure map') such that $\xi\circ \eta_X = \id_X$ and
$\xi\circ\mu = \xi \circ T\xi$. A homomorphism of $\TT$-algebras $(X,\xi) \to (Y,\eta)$
is a morphism $\phi : X \to Y$ such that $\eta\circ T\phi = \phi\circ\xi$. The
category of $\TT$-algebras and homomorphisms is the \emph{Eilenberg-Moore category}
of $\sK$, denoted $\sK^\TT$. Note that for any $X$ in $\sK$, there is a natural
$\TT$-algebra structure on $TX$ given by $\eta_X$. This is the `free' $\TT$-algebra
`generated' by $X$: indeed, the functor $\sK \to \sK^\TT$ sending $X$ to $(TX,\mu_X)$
is left adjoint to the obvious forgetful functor. Finally, given a monad $\TT$ on
$\sK$ and a monad $\TT'$
on another category $\sK'$, a \emph{morphism of monads}
from $\TT$ to $\TT'$ is a functor $\Phi : \sK \to \sK'$ together with a
natural transformation $\phi : T'\Phi \to T$ compatible with the monad structures
in a rather obvious sense. The pair $(\Phi,\phi)$ induces a functor
$\sK^\TT \to \sK'^{\TT'}$ of Eilenberg-Moore algebras, where
given a $\TT$-algebra $(X,\xi)$ in $\sK$, the object $\Phi X$ becomes
a $\TT'$-algebra via the composite $\Phi\xi \circ\phi_X : T'\Phi X \to \Phi X$.

Consider the \emph{internal ideal completion} functor $\Idl : \Ord_\sE \to \Ord_\sE$.
Its action on objects sends $X$ to the sub-object of $\Omega^X$ parameterising
order ideals: i.e., a morphism $\phi : T \to \Omega^X$ factors through
$\Idl X$ if the corresponding sub-object $\Phi \subset T \times X$ satisfies:
\begin{enumerate}
\item (lower) $\id_T \times \pr_1 : (\le_X) \times_{X} \Phi \to T \times X$
(fibre product using $\pr_2$) factors through $\Phi$,
\item (directed) for any $\alpha : T' \to T$
and $x_1,x_2 : T' \to \Phi$ over $T$, there is
an epi $\beta : T'' \to T'$ and $x : T'' \to \Phi$
such that $\beta^*x_1, \beta^*x_2 \le x$.
\end{enumerate}
The partial order on $\Idl X$ is induced from $\Omega^X$.
The action on morphisms sends an isotone $\phi : X \to Y$
to the restriction of the direct image morphism $\phi_* : \Omega^X \to \Omega^Y$
to $\Idl X \to \Idl Y$. There is a natural transformation
$\eta:\id \to \Idl$ such that $\eta_X : X \to \Idl X \subset \Omega^X$
is adjoint to the characteristic morphism $X \times X \to \Omega$
classifying the transpose of $(\le_X)$. (For $\sE=\Set$, this map
sends $x \in X$ to the principal ideal $\dn x$.) We also have a natural transformation
$\mu : \Idl^2 \to \Idl$ such that $\mu_X : \Idl^2 X \to \Idl X$ is a restriction
of the `union' morphism $\Omega^{\Omega^X} \to \Omega^X$.

One checks that (1) $\IIdl=(\Idl,\mu,\eta)$ is a monad on $\Ord_\sE$,
(2) if a partially ordered object $X$ is an internal distributive lattice,
then so is $\Idl X$, and furthermore $\mu_X$ and $\eta_X$ are homomorphisms
of distributive lattices, so that $(\Idl,\mu,\eta)$ defines a monad $\IIdl'$ on $\dLat_\sE$.
Their Eilenberg-Moore categories are, respectively, the internal dcpo and
internal frames: $\Ord_\sE^\IIdl = \Dcpo_\sE$ and
$\dLat_\sE^{\IIdl'} = \Frm_\sE$. (For $\sE=\Set$, these are the usual
categories of dcpo, resp. frames, with Scott-continuous maps as morphisms.)

In a Kripke topos $\Set^\sC$, the internal ideal completion
$\Idl S$ of a functor $S : \sC \to \Ord$ sends
$c \in \sC$ to the poset of sub-functors $J \subset S|_{\up c}$
such that $J(c')$ is an order ideal in $S(c')$ for all $c'\ge c$ in $\sC$.
In particular, for each $c \in \sC$ and each $s \in S(c)$
we may consider the sub-functor $\dn s \subset S$ sending
$c' \in \sC$ to $\dn S(c\to c')s$ if $c' \ge c$ and
to $\{0\}$ otherwise. This yields an isotone map $S(c) \to f_* \Idl S$
where $f : \Set^\sC \to \Set$ is the usual geometric morphism.

\subsection{}
Let $f : \sE \to \sE'$ be a geometric morphism. We lift both $f^*$ and $f_*$ to
an adjoint pair of functors $f_* : \Ord_\sE \rightleftarrows \Ord_{\sE'} : f^*$,
using the fact that both preserve finite limits.
Given $(X, \le_X)$ in $\Ord_\sE$, we let $f_*$ send it
to $(f_*X, f_*(\le_X))$.
Given $(X', \le_{X'})$ in $\Ord_{\sE'}$
we let $f^*$ send it to $(f^*, f^*(\le_{X'}))$.
Furthermore, these functors restrict to functors between
the sub-categories of internal distributive lattices
and internal Boolean algebras.

The situation is more complicated for internal frames and internal dcpo.
Let us first observe that the canonical morphism $\Omega_{\sE'}^{f_* X}
\to f_* \Omega_{\sE}^X$ induces a natural transformation
$\Idl_{\sE'} f_* \to f_* \Idl_{\sE}$ of functors $\Ord_\sE \to \Ord_{\sE'}$,
inducing a morphism of monads
from $\Ord_\sE$ to $\Ord_{\sE'}$.
It follows that, given an $\IIdl_\sE$-algebra $X$ in $\Ord_\sE$, we
may turn $f_*X$ into an $\IIdl_{\sE'}$-algebra in $\Ord_{\sE'}$ with a composite
structure morphism
$$\Idl_{\sE'} f_* X \to f_* \Idl_\sE X
\to f_* X.$$
An analogous construction equips the direct image of an $\IIdl'_\sE$-algebra
in $\dLat_{\sE}$ with the structure of an $\IIdl'_{\sE'}$-algebra in
$\dLat_{\sE'}$. Hence, the direct image functor lifts to
$f_*:\Dcpo_\sE \to \Dcpo_{\sE'}$
and $f_*:\Frm_\sE\to
 \Frm_{\sE'}$.
The inverse image functor does not have this property.
However, there does exist a left adjoint to the direct image functor on
dcpo, resp. frames, denoted $f^\sharp : \Dcpo_{\sE'} \to \Dcpo_\sE$,
resp. $f^\sharp : \Frm_{\sE'} \to \Frm_\sE$.

\subsection{}
A \emph{natural numbers object} in a general topos $\sE$ is
an object $\NN_\sE$ equipped with morphisms $0:1_\sE \to \NN_\sE$ (`zero') and
$S:\NN_\sE \to \NN_\sE$ (`successor map'),
satisfying a universal property: for every
$X$ equipped with morphisms $x : 1_\sE \to X$ and $\varphi : X \to X$
there is a unique morphism $\Phi : \NN_\sE \to X$ such that $\Phi\circ 0 = x$
and $\Phi\circ S = \varphi \circ \Phi$. The morphism $\Phi$ describes
the sequence of images of $x$ under subsequent iterates of $\varphi$.
In particular, iterating the composition-with-$S$
as an endomorphism of $\NN_\sE^{\NN_\sE}$ gives rise to
an addition $\NN_\sE \times \NN_\sE \to \NN_\sE$.
This makes $\NN_\sE$ (if it exists) into
a monoid object, and we may
proceed to construct the integers $\ZZ_\sE$ as a quotient of
$\NN_\sE \times \NN_\sE$ by the diagonal action of $\NN_\sE$.
The sub-object $0:1_\sE \to \ZZ_\sE$ admits a
complement $\ZZ_\sE^\times$; the latter is a monoid object
for `multiplication' and we obtain the rationals
$\QQ_\sE$ as a quotient of $\ZZ_\sE \times \ZZ_\sE^\times$ by the the diagonal
action of $\ZZ_\sE^\times$.
If $\sE$ admits a geometric morphism $f : \sE \to \Set$, we have simply
$\NN_\sE \simeq f^*\NN$, $\ZZ_\sE \simeq f^*\ZZ$ and $\QQ_\sE \simeq f^*\QQ$.
In particular, in a Kripke topos $\Set^\sC$ the natural numbers, integers and rationals
are simply the corresponding constant functors.

The construction of the reals is more involved.
Observe first that the objects of
natural numbers, integers and rationals are naturally partially ordered
(and their inclusions are isotone). We may consider the
\emph{lower reals} $\RR_{\ell,\sE}$ as a sub-object of $\Omega^{\QQ_\sE}$
such that a morphism $\varphi : T \to \Omega^{\QQ_\sE}$ factors through
$\RR_{\ell,\sE}$ if and only if the corresponding sub-object
$\Phi \subset T \times \QQ_\sE$ satisfies:
\begin{enumerate}
\item (lower) $\id_T \times \pr_1 : (\le_{\QQ_\sE}) \times_{\QQ_\sE} \Phi \to T \times \QQ_\sE$
(fibre product using $\pr_2$) factors through $\Phi$,
\item (rounded) for each $q' : T' \to \Phi$
there is an epimorphism $\beta:T'' \to T'$ and morphism $q'' : T'' \to \Phi$
such that $\beta^*q' \le q''$ and
the equaliser of these two is  $0_\sE$,
\item (epic) the projection $\Phi \to T$ is epic,
\item (bounded) there is an epi $\alpha:T' \to T$
and a morphism $b : T' \to \QQ_\sE$ such that
$\alpha^*\Phi \le b$ in the obvious sense.
\end{enumerate}
Without the last two conditions, we would obtain the \emph{extended} lower reals.
Finally, the \emph{Dedekind reals} $\RR_{d,\sE}$ may be constructed as a sub-object
of $\RR_{\ell,\sE} \times \Omega^{\QQ_\sE}$ parameterising pairs consisting of a
lower real and a complement. More precisely, a morphism
$\varphi : T \to \RR_{\ell,\sE}\times\Omega^{\QQ_\sE}$ is a Dedekind real
if its two components $\varphi_1 : T \to \RR_{\ell,\sE} \subset \Omega^{\QQ_\sE}$ and
$\varphi_2 : T \to \Omega^{\QQ_\sE}$ define complementary sub-objects
of $T \times \QQ_{\sE}$. The Dedekind reals are an honest ring object
in $\sE$; the lower reals are an additive monoid object, and carry an
action of the multiplicative monoid object of positive lower reals. They are also
naturally partially ordered in a compatible way.

These constructions have a simple interpretation in a Kripke topos
$\Set^\sC$. The lower reals here are `context-dependent' inhabited, bounded, rounded lower
subsets of $\QQ$ that `grow' as the context is being refined: that is,
the functor $\RR_\ell$ maps $c \in \sC$ to the set of isotone
maps from $\up c$ to the (usual) reals. On the other hand, the Dedekind reals
are the constant functor $f^*\RR$ (for the geometric morphism $f : \Set^\sC \to \Set$).
Indeed, a Dedekind real consists of a lower real and a complement -- the latter
is also a `context-dependent' subset that `grows', so that in effect both halves
of the Dedekind cut are forced to remain (locally) constant.

\section{Box presentations and framed topoi}

\subsection{Box worlds}
A \emph{box world} is an abstraction of an experiment in which a single system is shared
among a number of parties. Each party has a collection of yes--no `questions', i.e. binary
measurements, it may subject the system to. It is assumed that during a single run of the experiment,
the system had been prepared in a certain state, and each party chooses and asks
at most one `question', and that the order in which the parties perform their measurements
is irrelevant (the latter is usually expressed by saying that the parties are causally separated).
This procedure is then repeated, with the system consistently being prepared in the same state
prior to observation. Each party records its choice of a question and the answer obtained,
and eventually the records are compiled into a list of pairs, each consisting of
a sequence of questions (`context') and a sequence of answers (`outcome').
Given sufficiently many records, and a fixed context, empirical probabilities are assigned
to the possible outcomes.

We formalise this situation as a map $S \to I$ of \emph{finite} sets, alternatively viewed as a family
$(S_i)_{i\in I}$ of fibres. The set $I$ indexes the parties, while $S_i$ is the collection
of questions available to the $i$-th party.
A \emph{context} is then a \emph{partial section},
i.e. a subset $c \subset S$ such that the projection
$c \to I$ is injective.
The set of contexts $\sC_{S/I}$ is then a subset
of the power-set of $S$, and thus inherits a partial order.
We may also view $\sC_{S/I}$ as a sub-category
of $\Set$. In particular,
restricting the free Boolean algebra functor $\Set \to \Bool$ to $\sC_{S/I}$
yields the functor
$$L_{S/I}  : \sC_{S/I} \to \Bool. $$
The latter is an internal Boolean algebra in the Kripke topos
$$\sE_{S/I} = \Set^{\sC_{S/I}}. $$
Finally, applying the internal ideal completion
viewed as a free frame functor $\Idl : \dLat_{\sE_{S/I}} \to \Frm_{\sE_{S/I}}$
yields an internal frame
$$ F_{S/I} = \Idl L_{S/I}. $$

\subsection{Functoriality}
We have thus associated with $S/I$ the pair $(\sE_{S/I}, F_{S/I})$ of
a topos and an internal frame. We shall now make it functorial.

\begin{defn}
A \emph{framed topos} is a pair $(\sE, F)$ where $\sE$ is a topos and $F$ is an internal
frame in $\sE$. A \emph{homomorphism of framed topoi} from $(\sE,F)$ to $(\sE', F')$ is a
geometric morphism $f : \sE \to \sE'$ together with a homomorphism $f^\sharp F' \to F$
of internal frames in $\sE$. The category of \emph{framed topoi} is denoted $\FrmTop$.
\end{defn}

\begin{defn}
The category $\BBox$ of \emph{box presentations}
is a sub-category of the arrow category $\fSet^\to$
with maps $S \to I$ of finite sets as objects, and
commutative diagrams
$$ \begin{CD}
S' @>>> S \\
@VVV @VVV \\
I' @>>> I
\end{CD} $$
such that $S'_{i'} \to S$ is injective for all $i'\in I'$ as morphisms.
\end{defn}

\begin{pro}\label{pro:functor}
The assignment sending $S/I$ to $(\sE_{S/I}, F_{S/I})$ is the object part of
a functor $\BBox^\op \to \FrmTop$.
\end{pro}
\begin{proof}
Suppose given a morphism in $\BBox$, i.e. a commutative diagram
$$ \begin{CD}
S' @>{\varphi}>> S \\
@VVV @VVV \\
I' @>{\bar\varphi}>> I
\end{CD} $$
of finite sets, such that $S'_{i'} \to S$
is injective for all $i'\in I$.
We obtain an isotone map
$\phi:\sC_{S/I} \to \sC_{S'/I'}$ as a restriction
of the inverse image map $2^{S'} \to 2^{S}$ on power-sets.
In fact, viewing $\phi$ as a functor between sub-categories of $\Set$,
we also have a natural transformation $\tilde\varphi: \phi^*\iota'\to\iota$
between the embedding functors $\iota : \sC_{S/I} \to \Set$
and $\iota' : \sC_{S'/I'} \to \Set$. Letting $B: \Set\to\Bool$ be the
free Boolean algebra functor, we then have a natural transformation
$B\tilde\varphi : \phi^* B\iota'\to B\iota $ of functors $\sC_{S/I} \to \Bool$.
That is nothing but a homomorphism $ \phi^* L_{S'/I'} \to L_{S/I}$ of internal
Boolean algebras in $\Set^{\sC_{S/I}}$, where we use $\phi : \Set^{\sC_{S/I}}
\to\Set^{\sC_{S'/I'}}$ to denote the geometric morphism induced by the map on posets.
Passing to the adjoint $L_{S'/I'} \to \phi_*L_{S/I}$, applying $\Idl'$ and
using the natural transformation $\Idl'_{\sE_{S'/I'}} \phi_* \to \phi_* \Idl'_{\sE_{S/I}}$
we have the composite
$$
F_{S'/I'} = \Idl' L_{S'/I'} \to \Idl' \phi_* L_{S/I} \to \phi_* \Idl' L_{S/I} = \phi_* F_{S/I}.
$$
It is a homomorphism of internal frames in $\sE_{S'/I'}$ by construction of the
$\IIdl'$-algebra structure on $\phi_* F_{S/I}$. By adjunction, we obtain an internal
frame homomorphism $\phi^\sharp F_{S'/I'} \to F_{S/I}$ in $\sE_{S/I}$.
Hence, $(\phi, \phi^\sharp)$ is a morphism in $\FrmTop$
between $(\Set^{\sC_{S/I}}, F_{S/I})$ and $(\Set^{\sC_{S'/I'}}, F_{S'/I'})$: the
value of the desired functor on the morphism $(\varphi,\bar\varphi)$ in $\BBox$ between
$S/I$ and $S'/I'$.
A careful inspection of all the constructions involved shows that composition of morphisms
in $\BBox$ is compatible with composition of morphisms in $\FrmTop$.
\end{proof}

\subsection{Products}

Observe that the category $\BBox$ has finite coproducts:
indeed, $\emptyset\to\emptyset$ is an initial object, while
the coproduct of $S/I$ and $T/J$ is $(S \sqcup T) / (I \sqcup J)$.
On the other hand, the category $\FrmTop$ does not have all finite products,
and in particular it lacks a final object.
It is convenient to restrict to its full subcategory $\bFrmTop$
of \emph{bounded} framed topoi, consisting of those $(\sE,F)$ for which
there exists a geometric morphism $f:\sE \to \Set$ and an object $B$ in $\sE$
such that every object of $\sE$ is a sub-quotient of $B \times f^*A$ for some
set $A$. Now, $\bFrmTop$ does have finite products: indeed,
$(\Set, 2)$ is the final object, while the product of $(\sE,E)$
and $(\sF, F)$ is $(\sE \times \sF, p_1^\sharp E \otimes p_2^\sharp F)$
where
$$ \sE \xleftarrow{p_1} \sE\times\sF \xrightarrow{p_2} \sF $$
is the product in the category of bounded topoi and geometric morphisms,
while $\otimes$ is the coproduct in $\Frm_{\sE\times\sF}$. It is easy to
compute in our case.
\begin{lem}\label{lem:products}
Let $\sE, \sF$ be bounded topoi, and let $L$ be an internal distributive lattice in $\sE$,
and $M$ an internal distributive lattice in $\sF$. Then the product of
$(\sE, \Idl' L)$ and $(\sF, \Idl' M)$ in $\bFrmTop$ is naturally isomorphic to
$(\sE\times \sF, \Idl'(p_1^*L\otimes p_2^*M))$ where $\otimes$ is the coproduct in
$\dLat_{\sE\times\sF}$, together with projections
$$ (\sE, \Idl' L) \xleftarrow{(p_1,p_1^\sharp)} (\sE\times\sF, \Idl' (p_1^*L \otimes p_2^*M))
\xrightarrow{(p_2,p_2^\sharp)} (\sF, \Idl' M) $$
in $\bFrmTop$ in which the frame homomorphisms
$$ p_1^\sharp \Idl' L \xrightarrow{p_1^\sharp} \Idl' (p_1^* L \otimes p_2^* M)
\xleftarrow{p_2^\sharp} p_2^\sharp \Idl' M $$
are induced, using adjunctions, by distributive lattice homomorphisms
$$p_1^* L \to p_1^*L \otimes p_2^*M \leftarrow p_2^*M$$
arising from the universal property of the coproduct.
\end{lem}
\begin{proof}
Let $(\sG, F)$ be an object of $\bFrmTop$ together with
morphisms
$$
(\sE,\Idl' L) \xleftarrow{(s,s^\sharp)} (\sG,F) \xrightarrow{(t,t^\sharp)} (\sF,\Idl' M).
$$
The geometric morphisms $s$ and $t$ factor uniquely through $\langle s,t\rangle : \sG
\to \sE\times\sF$. The homomorphisms of internal frames
$$ p^\sharp \Idl' L \rightarrow F \leftarrow q^\sharp \Idl' M $$
correspond by adjunction to
$$ \Idl' L \to s_* F,\quad t_*F \leftarrow \Idl' M $$
and thus to
$$ L \to s_* F,\quad t_*F \leftarrow M $$
in $\dLat_\sE$ and $\dLat_\sF$. These pull back to
$$
p_1^* L \to p_1^* p_{1*} \langle s,t\rangle_* F,\quad
p_2^*p_{2*} \langle s,t\rangle_* F \leftarrow p_2^* M
$$
in $\dLat_\sG$. Composing with the counits $p_1^*p_{1*}\to \id$ and
$p_2^*p_{2*}\to \id$ we do finally obtain a morphism
$$ p_1^* L \otimes p_2^* M \to \langle s,t\rangle_* F $$
in $\dLat_\sG$, extending by universality to a frame homomorphism
from $\Idl' (p_1^* L \otimes p_2^* M)$ and giving by adjunction a morphism
$$ \langle s,t\rangle^\sharp \Idl'(p_1^* L \otimes p_2^* M) \to F $$
in $\Frm_\sG$. We have thus produced a morphism
$$ \langle (s,s^\sharp), (t,t^\sharp) \rangle : (\sG,F) \to (\sE\times\sF,
\Idl'(p_1^*L\otimes q_2^*M)) $$
in $\bFrmTop$. By construction, its composites with $(p_1,p_1^\sharp)$ and $(p_2,p_2^\sharp)$
give $(s,s^\sharp)$ and $(t,t^\sharp)$. Uniqueness follows from the universal property
of the coproduct $\otimes$ in $\dLat_\sG$.
\end{proof}

\begin{pro}
The functor of Proposition \ref{pro:functor} factors through
a functor $\BBox^\op \to \bFrmTop$ preserving finite products.
\end{pro}
\begin{proof}
Note first that the final object $\emptyset/\emptyset$
of $\BBox^\op$ is sent to the final object $(\Set,2)$ of $\bFrmTop$.
Suppose now given
$S/I$ and $T/J$, and let $U = S \sqcup T$, $K = I \sqcup J$.
A context for $U/K$ is a subset of $S \sqcup T$
such that its intersection with $S$ is a context for $S/I$
and its intersection with $T$ is a context for $T/J$; hence
$$ \sC_{U/K} \simeq \sC_{S/I} \times \sC_{T/J} $$
as posets. Let
$$ \sC_{S/I} \xleftarrow{s} \sC_{U/K} \xrightarrow{t} \sC_{T/J} $$
be the two projections, and denote by
$\iota_S : \sC_{S/I} \to \Set$, $\iota_T : \sC_{T/I} \to \Set$,
$\iota_U : \sC_{U/K} \to \Set$ the
canonical embeddings as sub-categories of $\Set$. Observe
that
$$
\iota_U \simeq s^*\iota_S \sqcup t^*\iota_T\quad\textrm{in}\ \sE_{U/K}
$$
so that
$$
L_{U/K} = B\iota_U \simeq s^* B\iota_S \otimes t^* B\iota_T \simeq s^* L_{S/I} \otimes
t^* L_{T/J},
$$
where $\otimes$ is the coproduct in $\Bool_{\sE_{\sC_{U/K}}}$, given
point-wise by the tensor product of Boolean algebras in $\Set$, and
thus coinciding with the coproduct in $\dLat_{\sE_{\sC_{U/K}}}$.
Hence, by Lemma \ref{lem:products}, $(\sE_{U/K}, \sF_{U/K})$ is the product
of $(\sE_{S/I}, F_{S/I})$ and $(\sE_{T/J}, F_{T/J})$.
\end{proof}

\subsection{Example}
The \emph{gbit} is a box world presented by the unique map $2 \to 1$.
The Popescu-Rohrlich box world is then presented by the coproduct of a pair
of gbits in $\BBox$, i.e. a map $2 \sqcup 2 \to 1 \sqcup 1$. Accordingly,
the framed topos associated to the P-R box world is the product of two copies
of the framed topos $(\sE_{2/1}, F_{2/1})$ associated with a gbit.
For a general object $(\sE,F)$ of $\bFrmTop$,
it is interesting to consider \emph{epimorphisms} to the product of
$n$ copies of $(\sE_{2/1}, F_{2/1})$: the supremum of $n$ for which such an epimorphism exists
is a certain measure of \emph{information capacity} of a hypothetical system described
by $(\sE,F)$. This in particular applies to $(\sE_{S/I}, F_{S/I})$ for a given
$S/I$ in $\BBox$.

\subsection{Remark}
Given a box presentation $S/I$ and an object $(\sE, F)$ of $\FrmTop$,
to give a morphism $(\sE, F) \to (\sE_{S/I}, F_{S/I})$ is the same as
to give a geometric morphism $\phi:\sE \to \sE_{S/I}$ together with
a homomorphism $\varphi : L_{S/I} \to \phi_* F$ of \emph{internal distributive lattices}
in $\sE_{S/I}$. Furthermore, $\varphi$ factors through a homomorphism of
internal Boolean algebras into the internal sub-lattice of `complemented
elements' in $\phi_*F$. It thus seems somewhat spurious to consider
internal frames rather than simply internal Boolean algebras (or even
distributive lattices); this will be seen even more strongly when we characterise
states or probability valuations on $F_{S/I}$.
However, frames are a natural setting for the discussion of \emph{phase spaces}.
In any case, $L_{S/I}$ may be recovered from $F_{S/I}$ as, again, the
sub-lattice of `complemented elements'.

\section{States}

\subsection{Box world states}
Consider a box world presented by a map $S \to I$ of finite sets.
According to the interpretation as a series of `simultaneous' observations performed
by parties indexed by $I$, the most general form of an `empirical' probability distribution
is a map, assigning a certain probability distribution to the possible outcomes in each context.
The condition of \emph{non-signalling} requires that the distribution
in a `partial' context (where some parties do not perform a measurement)
arise consistently as a marginal of the distribution in a `total' context refining
the partial one.

This motivates the following definition. Recall first that for an object $S/I$
of $\BBox$, we consider the elements of $\sC_{S/I}$ as subsets of $S$. Denote
by $\sC^m_{S/I} \subset \sC_{S/I}$ the set of \emph{maximal} elements.
\begin{defn}
A \emph{box world state} on $S/I$ is a map
$$ p : \coprod_{c \in \sC^m_{S/I}} 2^c \to [0,1] $$
such that
\begin{enumerate}
\item (normalisation)
$\sum_{x \in 2^c} p(x) = 1$
for each $c\in\sC^m_{S/I}$,
\item (non-signalling) for each $c \in \sC_{S/I}$ and $x \in 2^c$
the value
$$
p(x) := \sum_{\substack{ y \in 2^d \\ y|_c = x}} p(y)
$$
does not depend on the choice of $d \in \sC^m_{S/I}$, $c \le d$.
\end{enumerate}
\end{defn}

Recall that embedding $\sC_{S/I} \to \Set$, composing with the free Boolean algebra
functor and the forgetful functor
gives the object $L_{S/I}$ in $\sE_{S/I} = \Set^\sC$.
We may view it as a $\sC_{S/I}$-shaped diagram in $\Set$, and take its colimit
$\colim L_{S/I}$ (this is \emph{not} a colimit in $\Bool$). Observe that for each
$c \in \sC_{S/I}$ there are natural maps
$$
2^c \xrightarrow{\delta_c} 2^{2^c} \simeq L_{S/I}(c) \xrightarrow{\iota_c} \colim L_{S/I}
$$
where $\delta_c$ takes $x \in 2^c$ to its indicator function $\delta_c(x):2^c \to 2$.
The following will soon be useful.
\begin{lem}\label{lem:states}
The map
$$
\coprod_{c\in\sC^m_{S/I}} 2^c \xrightarrow{\iota\circ\delta} \colim L_{S/I}
$$
induces a one-to-one correspondence between
box-world states on $S/I$ and maps
$$ \rho : \colim L_{S/I} \to [0,1] $$
such that $\rho\circ \iota_c$ is a probability valuation
on the Boolean algebra $L_{S/I}(c)$
for each $c \in \sC^m_{S/I}$.
\end{lem}
\begin{proof}
Note first that for each finite set $A$, the set of probability
valuations on the power-set Boolean algebra $2^A$ is in natural bijection
with the set of functions $q: A \to [0,1]$ such that $\sum_{a\in A} q(a) = 1$.

Suppose now $p$ is a box world state on $S/I$.
It follows that
$p$ induces a probability valuation $p_c$ on $L_{S/I}(c)$ for each $c \in \sC_{S/I}^m$.
Furthermore, by the non-signalling condition, $p$ may be consistently extended to $2^{c'}$ for
each non-maximal $c'$, and thus defines a probability valuation $p_{c'}$ on $L_{S/I}(c')$.
By construction, $p_{c'} = p_c \circ L_{S/I}(c'\to c)$ whenever $c' \le c$ so that
$\coprod_c p_c : \coprod_c L_{S/I}(c) \to [0,1]$ descends to a
\emph{unique} map $\rho : \colim L_{S/I} \to [0,1]$ such that  $\rho \circ \iota_c = p_c$
for all $c\in\sC_{S/I}$.
In fact, since for every $c' \in \sC_{S/I}$ there exists $c \in \sC^m_{S/I}$ with
$c'\le c$, it follows that $\rho$ is the \emph{unique} $[0,1]$-valued map
from $\colim L_{S/I}$ such that $\rho\circ\iota_c = p_c$ for all $c\in\sC_{S/I}^m$.

Conversely, it is immediate that given a map $\rho : \colim_{S/I} \to [0,1]$
such that $\rho\circ\iota_c$ is a probablity valuation for all $c\in\sC^m_{S/I}$,
and pulling back along $\iota\circ\delta$, one obtains a box-world state
$p : \coprod_{c\in\sC^m_{S/I}} 2^c \to [0,1]$. Applying the construction of the
previous paragraph to $p$ gives $\rho$, whence the one-to-one correspondence.
\end{proof}

\subsection{Internal states}
We shall now recall the definition of probability valuations
on internal distributive lattices and frames in a topos $\sE$. We
require that $\sE$ possess a natural numbers object $\NN_\sE$, so
that we may construct the lower reals object $\RR_{\ell,\sE}$.
Recall that the latter is an additive monoid and a partially ordered object.
We denote by $[0,1]_{\ell,\sE}$ the sub-object of $\RR_{\ell,\sE}$ such that
a morphism $r:T \to \RR_{\ell,\sE}$ factors through $[0,1]_{\ell,\sE}$
if and only if $0_T \le r \le 1_T$ (where $0_T, 1_T : T \to \RR_{\ell,\sE}$
are the obvious `constant' morphisms).
Now, $[0,1]_{\ell,\sE}$ is an internal dcpo, i.e. an $\IIdl_\sE$-algebra.
\begin{defn}
Let $L$ an internal distributive lattice in $\sE$. An \emph{probability valuation} on $L$
\emph{as a distributive lattice}
is an isotone morphism $P: L \to [0,1]_{\ell,\sE}$ satisfying the conditions of:
\begin{enumerate}
\item (normalisation) $P 0_L = 0_{[0,1]}$ and $P 1_L = 1_{[0,1]}$,
\item (modularity) the diagram
$$\begin{CD}
L \times L @>{P\times P}>> [0,1]^2_{\ell,\sE} @>{+}>> \RR_{\ell,\sE} \\
@V{\langle \vee, \wedge \rangle}VV @. @| \\
L \times L @>>{P\times P}> [0,1]^2_{\ell,\sE} @>>{+}> \RR_{\ell,\sE}
\end{CD}$$
commutes.
\end{enumerate}
If $L$ is an internal frame, $P$ is
a probability valuation on $L$ \emph{as a frame} if additionally it satisfies
\begin{enumerate}
\item[(3)] (continuity) the diagram
$$\begin{CD}
\Idl L @>{\Idl P}>> \Idl [0,1]_{\ell,\sE} \\
@VVV @VVV \\
L @>>{P}> [0,1]_{\ell,\sE}
\end{CD}$$
commutes, i.e. $P$ is a morphism of internal dcpo.
\end{enumerate}
\end{defn}

Note that $P$ is a morphism in $\Ord_\sE$,
so that $\Idl P$ is a well-defined.
The vertical arrows in the last diagram are the structure morphisms
for $\IIdl_\sE$-algebras (recall that a frame is in particular a dcpo).

\begin{lem}\label{lem:val-frm-lat}
Let $L$ be an internal distributive lattice in $\sE$. Then the natural morphism
$\eta_L : L \to \Idl' L$ induces a one-to-one correspondence between
probability valuations on $L$ (as a distributive lattice) and
probability valuations on $\Idl' L$ (as a frame).
\end{lem}
\begin{proof}
By adjunction, $\eta_L$ induces
a bijection
$$\Hom_{\Dcpo_{\sE}}(\Idl L, [0,1]_{\ell,\sE}) \simeq \Hom_{\Ord_{\sE}}(L, [0,1]_{\ell,\sE})$$
where we recall that $\Idl' L$ as an internal dcpo is the same as $\Idl L$. Hence an isotone
morphism $\Idl' L \to [0,1]_{\ell,\sE}$ satisfying conditions (1), (2), (3) of the Definition
is the same as an isotone morphism $L \to [0,1]_{\ell,\sE}$ satisfying conditions (1), (2).
\end{proof}

\begin{lem}\label{lem:val-bool}
Let $L$ be an internal Boolean algebra in $\sE$. Then every probability
valuation $P : L \to [0,1]_{\ell,\sE} \subset \RR_{\ell,\sE}$ factors through
the Dedekind reals $\RR_{d,\sE}$.
\end{lem}
\begin{proof}
Let $\nu : \RR_{d,\sE} \to \RR_{\ell,\sE}^2$ be the morphism taking a Dedekind cut
to its lower and negative upper part. We then have a pullback diagram
$$\begin{diagram}
\node{\RR_{d,sE}}\arrow{s}\arrow{e,t}{\nu}\node{\RR_{\ell,\sE}^2}\arrow{s,r}{+} \\
\node{1_\sE}\arrow{e,t}{0_\RR} \node{\RR_{\ell,\sE}.}
\end{diagram}$$
Now, consider the diagram
$$ L \xrightarrow{\langle \id,\neg\rangle}
L \times L \overset{\id}{\underset{\wedge \times \vee}{\rightrightarrows}}
L \times L \xrightarrow{P\times (P-1)}
\RR_\ell^2 \xrightarrow{+} \RR_{\ell,\sE} $$
where the two composites coincide by modularity of $P$, and are in fact
equal to the composite $$L \to 1_\sE \xrightarrow{0_\RR} \RR_{\ell,\sE}.$$
It follows that the top composite $L \to \RR^2_{\ell,\sE}$ factors through
$\nu$. Composing with the first projection $\RR^2_{\ell,\sE} \to \RR_{\ell,\sE}$,
we have that $P : L \to \RR_\ell$ factors through
the `lower part' map $\RR_{d,\sE} \to \RR_{\ell,\sE}$.
\end{proof}

\subsection{Main result}
\begin{lem}\label{lem:val-bool-2}
Let $L$ be an internal Boolean algebra in a Kripke topos $\sE = \Set^\sC$,
and $f : \sE \to \Set$ an essential
geometric morphism. Then for each probability valuation $P : L \to [0,1]_{\ell,\sE}$
there is a map $\bar P : f_! L \to [0,1]$ such that the diagram
$$\begin{CD}
L  @>{P}>> [0,1]_{\ell,\sE} \\
@VVV @AAA \\
f^*f_! L @>>{f^*\bar P}> f^*[0,1]
\end{CD}$$
commutes.
\end{lem}
\begin{proof}
By Lemma \ref{lem:val-bool}, $P$ factors through
$\RR_{d,\sE}$. Since $\sE$ is a Kripke topos, $\RR_{d,\sE}$
is isomorphic to $f^*\RR$. By adjunction $f_! \vdash f^*$ we have
$\bar P : f_! L \to \RR$ such that $f^* \bar P$ composed with the unit $L \to f^*f_!L$
gives $P$. Finally, since $P$ factors through $[0,1]_{\ell,\sE}$
and thus through $[0,1]_{d,\sE} \simeq f^* [0,1]$, it follows that $\bar P$
factors through $[0,1]$.
\end{proof}

\begin{thm}\label{thm:states}
Let $S/I$ be an object of $\BBox$. The are natural one-to-one correspondences between:
\begin{enumerate}
\item probability valuations on $F_{S/I}$ as an internal frame in $\sE_{S/I}$,
\item probability valuations on $L_{S/I}$ as an internal distributive lattice in $\sE_{S/I}$,
\item box world states on $S/I$.
\end{enumerate}
\end{thm}
\begin{proof}
The bijection between (1) and (2) follows from Lemma \ref{lem:val-frm-lat}.
Since $L_{S/I}$ is Boolean, probability valuations on $L_{S/I}$ as an internal distributive
lattice are in bijection with maps $\bar P: f_! L_{S/I} \to [0,1]$ such that
composing $f^*\bar P$ with $L \to f^*f_!L$ is a probability valuation (Lemma \ref{lem:val-bool-2}).
Since
$f_! L_{S/I} \simeq \colim L_{S/I}$, we then have a further bijection with maps
$\rho : \colim L_{S/I} \to [0,1]$ such that the lift $$\tilde \rho : \coprod_{c\in\sC_{S/I}}
L_{S/I} \to [0,1] $$ defines a probability valuation $L_{S/I} \to f^*[0,1]$.
The latter condition is equivalent to requiring that the restriction
of $\tilde\rho$ to $L_{S/I}(c)$ be a probability valuation for each $c\in\sC_{S/I}$.
In fact, since for each $c'\in\sC_{S/I}$ there is
a $c \in \sC^m_{S/I}$ with $c' \le c$, the above requirement need only be stated
for maximal contexts. Hence, by Lemma \ref{lem:states}, we have a bijection between (2) and (3).
\end{proof}

\section{Phase space}

\subsection{Frames and locales}
Internal locales in $\sE$
are just internal frames, with the direction of morphism reversed:
$\Loc_\sE = \Frm_\sE^\op$. Given a geometric morphism $f:\sE \to \sE'$,
the adjunction $f^\sharp \vdash f_*$ may be viewed as a pair of functors
between $\Loc_\sE$ and $\Loc_{\sE'}$. We use $\OO$ to denote the tautological contravariant
functor from frames to locales. Now, given a locale $X$ in $\Set$, the frame $\OO(X)$
viewed as a category comes with a Grothendieck topology where a sieve
$(u_i \to u)_{i\in I}$ is covering whenever $u = \bigvee_{i\in I} u_i$.
Thus $\OO(X)$ becomes a site, and its topos of sheaves is denoted $\Sh(X)$.
The functor $\Sh$ from locales to topoi is full and faithful, and its essential image
is the category of \emph{localic} topoi. These are in particular bounded,
with a geometric morphism $\Sh(X) \to \Set$ induced by
the unique frame homomorphism $2 \to \OO(X)$.
We observe that if $M$ is a topological space, its topology is a frame and
thus we may consider the corresponding locale: this gives a functor from
topological spaces to locales, admitting a right adjoint.
Stone duality
restricts it to an equivalence between the category of \emph{sober topological spaces}
with the category of \emph{spatial locales} in $\Loc$.
In any case, the usual topos of sheaves
on a topological space $M$ is canonically equivalent to the topos
of sheaves on the corresponding locale.

We let $\lFrmTop \subset \bFrmTop$ denote the full subcategory of localic framed topoi.
In particular, for $S/I$ in $\BBox$ we have $\sE_{S/I} \simeq \Sh(\sC_{S/I})$
where we view $\sC_{S/I}$ as a sober topological space using the
Alexandrov topology. Thus the functor $\BBox^\op \to \FrmTop$
factors through $\lFrmTop$.

Note that for each $(\Sh(X),F)$ in $\lFrmTop$ the internal frame $F$ may be viewed
as an internal locale.  The direct image functor $\Loc_{\Sh(X)} \to \Loc$
sends the final internal locale, corresponding to the initial frame $\Omega_{\Sh(X)}$,
to $X$ itself, and thus factors through the slice category $\Loc/X$.
A fundamental result in topos theory~\cite[Thm. C1.6.3]{johnstone} is that
this gives an equivalence of categories
$\Loc_{\Sh(X)} \simeq \Loc/X.$
Following the definitions, we thus obtain
an equivalence
$$
\lFrmTop \simeq \Loc^\to
$$
between localic framed topoi and the arrow category of locales in $\Set$.
Composing with the functor from box world presentations gives the
\emph{external phase space} functor
$$
\XX : \BBox^\op \to \Loc^\to.
$$
It turns out that the latter factors through the
arrow sub-category of \emph{spatial} locales,
and thus through the category of homeomorphisms
of sober topological spaces.
We will give its explicit description.

\subsection{The external phase space}
Consider a box world presented by $S \to I$. Recall once again that
$\sC_{S/I}$ is a subset of the power-set $2^S$, ordered by inclusion.
The functor $L_{S/I}$ takes a context $c \subset S$ to the free
Boolean algebra on $c$, and may be viewed as a functor into the category
$\fBool$ of finite Boolean algebras. Now, we have an adjunction
$$2^- : \Set \rightleftarrows \fBool^\op : \Sp $$
where the \emph{spectrum functor} is right-adjoint to
the powerset functor. Thus, the composite
$$ \Sp \circ L_{S/I} : \sC^\op_{S/I} \to \Set $$
is a \emph{contravariant} set-valued functor, i.e. a \emph{presheaf} on $\sC_{S/I}$.
We construct the \emph{category of elements}
$$ X_{S/I} = \int \Sp\circ L_{S/I}, $$
fibred over $\sC_{S/I}$. Its set of objects is the disjoint union
$\coprod_{c\in \sC_{S/I}}  \Sp L_{S/I}(c)$, while the hom-set
$\Hom_{X_{S/I}}(x,x')$ is a singleton if
$x$ is the image of $x'$ under $\Sp L_{S/I}(c' \to c)$, or empty otherwise.
Thus, $X_{S/I}$ is in fact a poset, with an isotone projection to $\sC_{S/I}$.
We may give a very straightforward description as follows: just as $\sC_{S/I}$
is the poset of \emph{contexts}, $X_{S/I}$ is the poset of \emph{outcomes},
ordered by refinement. That is:
$$ X_{S/I} \simeq \{ (c,x)\ |\ c \in \sC_{S/I},\ x:c \to 2 \} $$
where $(c',x') \le (c,x)$ if and only if $c' \le c$ and $x|_{c'}=x'$.

\begin{lem}\label{lem:ps-explicit}
The external phase space of $S/I$ is isomorphic, as an object of $\Loc^\to$, to the
continuous map of Alexandrov spaces $X_{S/I} \to \sC_{S/I}$.
\end{lem}
\begin{proof}
By construction, the external phase space is the homeomorphism of locales corresponding
to the frame homomorphism $f_* \Omega_{S/I} \to f_* F_{S/I}$
where $f : \sE_{S/I} \to \Set$ is the unique geometric morphism,
and $\Omega_{S/I}$ is the sub-object classifier in $\sE_{S/I}$. Since
the poset $\sC_{S/I}$ admits a lower bound $\emptyset$ (the empty context),
we have that
$$ f_* F = \lim F = F(\emptyset) $$
for each functor $F : \sC \to \Set$. We already know that
$f_*\Omega_{S/I} = \Omega_{S/I}(\emptyset)$ is the frame $\OO(\sC_{S/I})$
for the Alexandrov topology: indeed, it is the lattice of
subfunctors of $1_{\sE_{S/I}}$, coinciding with the lattice of
upper sets in $\sC_{S/I}$. On the other hand,
$f_* F_{S/I} = (\Idl L_{S/I})(\emptyset)$ is the lattice of
subfunctors $J \subset L_{S/I}$ such that $J(c) \subset L_{S/I}(c)$ is an
ideal for all $c \in \sC_{S/I}$. Since $L_{S/I}(c)$ is finite, every ideal
is principal, and thus $f_* F_{S/I}$ is naturally identified with the
poset of `sections' $u : \sC_{S/I} \to \coprod_c L_{S/I}(c)$ such that
$u(c) \ge L_{S/I}(c' \to c) u(c')$ whenever $c' \le c$.
With this latter description, the homomorphism $\OO(\sC_{S/I}) \to f_* F_{S/I}$
sends an upper set $U \subset \sC_{S/I}$ to a section $u$ such that
$u(c) = 1$ if $c \in U$, or $u(c)=0$ otherwise. Now, note that
$L_{S/I}(c)$ is precisely the powerset of the fibre of $X_{S/I}$ over $c$,
i.e. $2^{\Sp L_{S/I}(c)}$. Thus, a section $u$ in $f_* F_{S/I}$
is the same as a subset $G \subset X_{S/I}$ such that the \emph{preimage}
of the fibre $G_{c'} \subset \Sp L_{S/I}(c')$ under
$\Sp L_{S/I}(c'\to c)$ is contained in $G_{c} \subset \Sp L_{S/I}(c)$.
Hence, $f_* F_{S/I}$ is canonically identified with the lattice of upper sets in $X_{S/I}$,
i.e. the frame $\OO(X_{S/I})$ for the Alexandrov topology. Furthermore,
the homomorphism $\OO(\sC_{S/I}) \to f_*F_{S/I} \simeq \OO(X_{S/I})$ now takes
an upper set $U \subset \sC_{S/I}$ to its preimage in $X_{S/I}$ under the
natural projection, whence the corresponding morphism of locales
-- i.e., of Alexandrov spaces -- is the projection again.
\end{proof}

Recall that given a morphism $S'/I' \to S/I$
in $\BBox$, with underlying map $\varphi : S' \to S$, the corresponding
map on contexts $\sC_{S/I} \to \sC_{S'/I'}$ is given by the pullback
map $\varphi^* : 2^S \to 2^{S'}$. Now, given an outcome $x : c \to 2$
with $c \subset S$, we obtain $\varphi^*x : \varphi^{*}c \to 2$, an outcome in
the context $\varphi^{*}c$. It thus follows that the map on contexts lifts
to a map on the external phase spaces $X_{S/I} \to X_{S'/I'}$:
$$\begin{CD}
X_{S/I} @>{\varphi^*}>> X_{S'/I'} \\
@VVV @VVV \\
\sC_{S/I} @>>{\varphi^{*}}> \sC_{S'/I'}.
\end{CD}$$
All arrows in the above diagram are isotone, and thus
may be viewed as continuous maps between Alexandrov spaces.
Identifying the latter with spatial locales, we obtain a functor
$$ \tilde\XX : \BBox^\op \to \Loc^\to. $$
\begin{lem}
$\tilde\XX$ is naturally isomorphic to $\XX$.
\end{lem}
\begin{proof}
The components $\lambda_{S/I} : \XX(S/I) \to \tilde\XX(S/I)$ have already
been constructed in the proof of Lemma \ref{lem:ps-explicit}, as
diagrams $$ \begin{CD}
f_* \Omega_{S/I} @>{\simeq}>> \OO(\sC_{S/I}) \\
@VVV @VVV \\
f_* F_{S/I} @>>{\simeq}> \OO(X_{S/I})
\end{CD}$$
of frame homomorphisms, where $f : \sE_{S/I} \to \Set$ is the
geometric morphism. Checking that they
combine into a natural transformation is
a straightforward, if somewhat tedious, unwinding of the constructions
performed in the proof of Proposition \ref{pro:functor},
and comparison with the definition of $\tilde\XX$.
\end{proof}

\subsection{Examples}
The external phase space of the gbit is a five-point set
$X_{2/1}=\{*, a,a',b,b'\}$ where $*$ is the unique closed point,
and the
$\{a\},\{a'\},\{b\},\{b'\}$ form a basis for the topology
More precisely, $\XX(2/1)$ is
the projection $X_{2/1} \to \sC_{2/1}$ with fibres
$\{*\}$, $\{a,a'\}$, $\{b,b'\}$. Note that since
$\lFrmTop \to \Loc^\to$ is an equivalence, and
$\BBox^\op \to \lFrmTop$ preserves finite product,
it follows that so does $\XX$. Thus, the external phase
space of the PR box world is the product
$X_{2/1} \times X_{2/1} \to \sC_{2/1} \times \sC_{2/1}$.
This may of course be seen on a completely elementary level, viewing the
external phase space as a fibration of `outcomes' over `contexts'
(with topologies consisting of subsets stable under refinement).

\section{Discussion}
\label{sec:discussion}

\subsection{States and morphisms}
The assignment of the set of box world states to a presentation
$S/I$ in $\BBox$ gives rise to a functor $\BBox^\op \to\Conv$
into the category of convex spaces. It would be desirable to
extend such functoriality to more general framed topoi.
A possible version of this may be achieved if one works
with the \emph{locale} of probability valuations, rather
than a set. Given a bounded topos $\sE$,
Vickers~\cite{vickers2} defines a \emph{valuation monad} $\VV_\sE$
on $\Loc_\sE$.
It sends a locale corresponding to a frame $F$ to the locale
whose global points are precisely the probability
valuations on $F$. Furthermore, $\VV_\sE$-algebras are the internal \emph{convex locales},
so that the locale of probability valuations is tautologically convex. The most imporant
aspect of his construction is however geometricity, giving rise to a natural isomorphism
$f^\sharp V_\sE \to V_{\sE'} f^\sharp$ for a geometric morphism $f : \sE \to \sE'$.
It then follows that sending $(\sE, F)$ to $(\sE, V_\sE F)$ induces a functor
$\sV : \bFrmTop \to \bFrmTop$.

\subsection{Channels}
The most general notion of a channel between a pair of box worlds
is a map on their state spaces, preserving the convex structure. For
classical systems, this is equivalent to giving a probability kernel, or stochastic map,
between their phase spaces. In particular, a channel
from a system to itself is simply a Markov kernel, or a stochastic self-map.
Given a bounded framed topos $(\sE,F)$
one might by analogy consider homomorphisms $V_\sE F \to F$ of internal
frames (i.e. work in the Kleisli category of $\VV_\sE$), as `internal
channels' from the system to itself. However, these are forced to preserve
to implicit structure of contexts, and thus capture only the most `classical'
possibilities. For example, given a box world presented by $S/I$, a
homomorphism $V_{\sE_{S/I}} F_{S/I} \to F_{S/I}$ describes merely a
negation, with some probabilies, of some of the `propositions' in $S$. Thus,
a question remains how to give a proper description of channels, perhaps in
terms of $\sV$, extending the standard one for box worlds.

\subsection{General non-signalling theories}
\label{ss:general}
Our framework suggests a moderate generalisaion of non-signalling box worlds,
where the system is specified by a partially ordered set $\sC$ of contexts,
and a functor $L:\sC \to \fBool$ of measurement logics (or, equivalently, the
`spectral presheaf' $\Sp L : \sC^\op \to \Set$). A convenient assumption on $\sC$
is that its chains be finite. This representation is applicable to orthodox
quantum systems (with a finite-dimensional Hilbert space), quantum logics modelled
by orthomodular lattices (see~\cite{heunen} for the general notion of Bohrification
of models of an  algebraic theory with respect to a sub-theory), and possibly
more exotic structures such as the logic of properties induced by a convex state
space \`a la Mielnik~\cite{mielnik}.

\section{Acknowledgements}
The authors acknowledge partial suppoort of the John Templeton Foundation grant no. 43174.

\bibliographystyle{plain}
\bibliography{bns}

\end{document}